\documentclass[a4paper,abstract]{scrartcl}

\usepackage[english]{babel}
\usepackage[utf8x]{inputenc}
\usepackage[T1]{fontenc}

\usepackage{amsmath}
\usepackage{amsfonts}
\usepackage{amssymb}
\usepackage{amsthm}
\usepackage{mathtools}

\usepackage{cite}

\usepackage{enumerate}

\usepackage{authblk}

\usepackage{todonotes}

\usepackage[numbers]{natbib}
\usepackage{hyperref}
\hypersetup{ 
	colorlinks,
	linkcolor={blue!60!black},
	citecolor={blue!60!black},
	urlcolor={blue!60!black},
	linktoc=page}
	
\renewcommand{\d}{\mathrm{d}}

\newcommand{\D}{\mathrm{D}}
\newcommand{\cD}{\mathfrak{D}}

\newcommand{\cE}{\mathcal{E}}

\newcommand{\cJ}{\mathcal{J}}
\newcommand{\cL}{\mathcal{L}}

\newcommand{\der}[2]{\frac{\partial {#1}}{\partial {#2}}}
\newcommand{\var}[3]{\frac{\delta_{#1} {#2}}{\delta {#3}}}

\newcommand{\C}{\mathbb{C}}
\newcommand{\N}{\mathbb{N}}
\newcommand{\R}{\mathbb{R}}

\DeclareMathOperator{\pr}{pr}

\newtheorem{theorem}{Theorem}
\newtheorem{lemma}[theorem]{Lemma}

\theoremstyle{definition}

\newtheorem{example}{Example}

\begin{document}

\title{Variational symmetries and pluri-Lagrangian structures \\
for integrable hierarchies of PDEs}
\author{Matteo Petrera and Mats Vermeeren}
\affil{
Institut f\"ur Mathematik, MA 7-1 

Technische Universit\"at Berlin, Str. des 17. Juni 136,
10623 Berlin, Germany

\medskip
\texttt{petrera@math.tu-berlin.de}, \texttt{vermeeren@math.tu-berlin.de}
}
\date{}

\maketitle

\begin{abstract}
We investigate the relation between
pluri-Lagrangian hierarchies of $2$-dimensional partial differential equations and their variational symmetries. The aim is to generalize to the case of partial differential equations the recent findings in [Petrera, Suris. J.\@ Nonlinear Math.\@ Phys.\@ 24:sup1, 121--145 (2017)] for ordinary differential equations. We consider hierarchies of $2$-dimensional Lagrangian PDEs (many of which have a natural $(1+1)$-dimensional space-time interpretation) and show that if the flow of each PDE is a variational symmetry of all others, then there exists a pluri-Lagrangian 2-form for the hierarchy. The corresponding multi-time Euler-Lagrange equations coincide with the original system supplied with commuting evolutionary flows induced by the variational symmetries. 
\medskip

\smallskip

\end{abstract}
\section{Introduction}

In the last decade a variational perspective on integrable systems has emerged under the name of \emph{pluri-Lagrangian} systems (or \emph{Lagrangian multiform} systems). The theory was initiated in \cite{LN1} in the discrete setting, more specifically in the context of multidimensionally consistent lattice equations on a quadrilateral stencil, called \emph{quad equations}. Multidimensional consistency means that the equation can be imposed on all elementary squares in a higher-dimensional lattice without leading to contradictions. Analogous to commutativity of differential equations, multidimensional consistency is a key feature of integrability for difference equations.

In \cite{LN1} it was shown that the property of multi-dimensional consistency can be combined with a variational formulation for quad equations. Solutions of integrable quad equations are critical points of an action functional obtained by integrating a suitable discrete Lagrangian 2-form over an arbitrary 2-dimensional surface in a higher-dimensional lattice. If the 2-dimensional surface is a plane, we recover a traditional variational principle for a 2-dimensional difference equation where the action is the sum over a plane of evaluations of the Lagrange function. The pluri-Lagrangian property requires the action to be critical also when this plane is replaced by any other 2-dimensional discrete surface in a higher-dimensional lattice.
This remarkable property has been considered as a defining feature of integrability of 2-dimensional discrete equations \cite{LN1, LN2,LN3,ALN,XLN,BS,BPS2,BPS5} as well as in the 1-dimensional \cite{YLN,BPS1,BPS3} and 3-dimensional \cite{LNQ,BPS4} cases.

The pluri-Lagrangian property can also be formulated in the continuous case, where solutions of (hierarchies of) integrable 2-dimensional partial differential equations (PDEs) are critical points of an action functional obtained by integrating a differential 2-form over an arbitrary 2-dimensional surface in a higher-dimensional space. This variational principle has been proposed as a Lagrangian analogue of the existence of  Poisson-commuting Hamilton functions \cite{S2,SV,LN1,XLN}. As in the discrete case, it is not limited to Lagrangian 2-forms describing 2-dimensional PDEs. The corresponding variational principle where a Lagrangian 1-form is integrated over curves applies to integrable ordinary differential equations \cite{S1,petrera2017variational,YLN}. It is conjectured that also for $d > 2$ integrable hierarchies of $d$-dimensional integrable PDEs can be described by pluri-Lagrangian $d$-forms.

Thanks to these investigations a quite suggestive scenario has emerged: the pluri-Lagrangian structure is closely related (or even equivalent) to the integrability of the underlying system. This novel characterization of integrability applies to both ordinary differential (or difference) equations and partial differential (or difference) equations. 

In the recent paper \cite{petrera2017variational} a connection between the notions of pluri-Lagrangian structures and variational symmetries was proved in the context of classical mechanics. In particular, it was shown that the existence of commuting variational symmetries for a system of variational ordinary differential equations leads to a natural pluri-Lagrangian 1-form, whose multi-time Euler-Lagrange equations consist of the original system
and commuting flows corresponding to the variational symmetries. These findings confirmed, in the framework of classical mechanics, that a pluri-Lagrangian structure is hidden behind the existence of a sufficient number of variational symmetries (i.e., of integrals of motion thanks to Noether theorem).

In the present work we extend the above idea to the case of variational $2$-dimensional PDEs, thus generalizing the results of \cite{petrera2017variational} to the context of Lagrangian field theory with two independent variables. We consider hierarchies of variational PDEs where the flow of each PDE is a variational symmetry of the Lagrange functions of all other members of the hierarchy. Under this assumption, we show that there exists a pluri-Lagrangian 2-form for the hierarchy.

The paper is organized as follows.
In Section \ref{sec-Noether} we give a short overview of Lagrangian field theory, recalling some classical notions and definitions. In particular we will provide a formulation of the celebrated Noether theorem, which establishes the  relation between conservation laws and variational symmetries.
In Section \ref{sec-pluri} we review the notion of continuous $2$-dimensional pluri-Lagrangian systems. 
Section \ref{sec-construction} is devoted to new results. It will be proved that from a family of variational symmetries one can construct a pluri-Lagrangian structure.
The final Section \ref{sec-examples}
contains three examples which illustrate the theoretical results obtained in Section \ref{sec-construction}.

\section{A short review of Lagrangian field theory}
\label{sec-Noether}

An exhaustive reference on classical Lagrangian field theory is the book of  P.J.\@ Olver \cite{Olver}. The scope of the present Section is to
 recall the main definitions and concepts needed for a self-contained presentation of our results in the next Sections.

\subsection{Euler-Lagrange equations}

Since we will work in a multi-time setting we do not restrict our presentation here to fields depending on only two independent variables. Therefore we start by considering
a smooth field $u: \R^N \to \R$ depending on $N$ real independent variables 
$t_1,\dots, t_N$. 

We will use the multi-index notation for partial derivatives. For any multi-index $I=(i_1,\ldots,i_N) \in \N^N$ we set
\[ 
u_I = \frac{\partial^{|I|} u}{(\partial t_1)^{i_1} \ldots (\partial t_N)^{i_N}},
\]
where $|I| = i_1 + \ldots + i_N$ and $u = u(t_1,\ldots,t_N)$. The notations
$It_k$ and $It_k^\alpha$ will represent the multi-indices $(i_1,\ldots,i_k + 1, \ldots i_N)$ and $(i_1,\ldots,i_k + \alpha, \ldots i_N)$ respectively. We will write $k \not\in I$ if $i_k = 0$ and $k \in I$ if $i_k \neq 0$. 

We will denote by $\D_i$ the total derivative with respect to the coordinate direction $t_i$,
\[ 
\D_i = \sum_{I \in \N^N} u_{I t_i} \der{}{u_I}
\]
and by $\D_I = \D_{1}^{i_1} \ldots \D_{N}^{i_N}$ the corresponding higher order derivatives.

The field $u$ can be considered as a section of the trivial bundle $\R^N \times \R$. The partial derivatives of $u$ of any order span the infinite jet bundle associated with $\R^N \times \R$. We will denote the fiber of the infinite jet bundle by $\cJ^\infty$ and the fiber coordinates by \[[u]=(u,u_{t_i},u_{t_it_j},\ldots)_{i,j,\ldots \in \{1,\ldots,N\}}.\]

A {\emph {variational problem}} for a smooth field $u: \R^N \rightarrow \R$ is described by a {\emph {Lagrangian}} $L: \cJ^\infty \rightarrow \R$ and consists
in finding the critical points of the
{\emph {action functional}}
\[
S = \int_\Gamma L[u] \, \d t_1 \wedge \cdots \wedge \d t_N ,
\]
where $\Gamma \subset \R^N$ is some bounded region. In other words, we look for fields $u$ such that for all fields $v$ such that $v$ and its derivatives vanish at the boundary of $\Gamma$, there holds
\[ \frac{\d}{\d \varepsilon}\bigg|_{\varepsilon = 0} \int_\Gamma L[u + \varepsilon v]\, \d t_1 \wedge \cdots \wedge \d t_N  = 0. \]

Concretely, we will be interested in variational problems for fields $u: \R^2 \rightarrow \R$. Therefore, let us fix $N=2$ and write explicitly the variational equations governing the evolution of $u$. In this case
the action functional over some bounded region $\Gamma \subset \R^2$ is
\begin{equation} \label{action}
S = \int_\Gamma L[u] \, \d t_1 \wedge \d t_2 .
\end{equation}
The field $u$ is a solution to the variational problem, i.e.,\@ a critical point for the action $S$, if and only if
\begin{equation} \label{varpde}
\var{}{L}{u} = 
\sum_{\alpha,\beta \geq 0} (-1)^{\alpha+\beta} \D_1^\alpha \D_2^\beta \!\left( \der{L}{u_{t_1^\alpha t_2^\beta}} \right)=
0,
\end{equation}
where the left hand side is called the \emph{variational derivative} of $L$.
Equation (\ref{varpde}) gives rise to a variational PDE, called \emph{Euler-Lagrange equation}. Note that if the Lagrangian depends on the $n$-th order jet, i.e.,\@ on derivatives of $u$ up to order $n$, then the Euler-Lagrange equation depends on the jet of order $2n$. If a given $2$-dimensional PDE can be written as in Equation (\ref{varpde}) for some Lagrangian $L$, then we say that this PDE has a variational (or Lagrangian) structure.

Of course, the Euler-Lagrange equation (\ref{varpde}) admits a straightforward generalization for the case of a field $u: \R^N \rightarrow \R$ for $N > 2$.

\begin{example}\label{ex-KdV-Lag}
The Korteweg-de Vries (KdV) equation
\[ w_2 = w_{111} + 6 w w_1, \]
where $w_i$ is shorthand notation for the derivative $w_{t_i}$,
can be put into a variational form by introducing the potential $u = w_1$. The corresponding equation is
\[ u_{12} = u_{1111} + 6 u_1 u_{11}. \]
Its variational structure comes from the Lagrangian
\[
L[u]= \frac{1}{2} u_1 u_2 - u_1^3 - \frac{1}{2} u_1 u_{111}.
\]
Indeed, critical points of the action (\ref{action})
are characterized by the Euler-Lagrange equation
\begin{align*}
0 = \frac{\delta L}{\delta u} 
&= - \D_1 \der{L}{u_1} - \D_1^3 \der{L}{u_{111}}
- \D_2 \der{L}{u_2} \\
&= \left( - \frac{1}{2} u_{12} + 6 u_1 u_{11} + \frac{1}{2} u_{1111}\right) + \frac{1}{2} u_{1111} - \frac{1}{2} u_{12} \\
&=-u_{12} + u_{1111} + 6 u_1 u_{11}.
\end{align*}
\end{example}

\subsection{Variational symmetries and Noether's theorem}

Let $N=2$.
A vertical generalized vector field on $\R^2 \times \R$ is a vector field of the form $Q \partial_u$, where $Q: \cJ^\infty \rightarrow \R$. It is called \emph{vertical} because it does not contain any $\partial_{t_i}$ and \emph{generalized} because $Q$ depends on derivatives of $u$, not just on $u$ itself. The {\emph {prolongation}} of $Q \partial_u$ is a  vector field on $\cJ^\infty$ defined as
\[ \pr (Q \partial_u) = \sum_{I \in \N^2} (\D_I Q) \der{}{u_I}. \]

A vector field $Q \partial_u$ is called a {\emph {variational symmetry}} of a Lagrangian $L:\cJ^\infty \rightarrow \R$
if its prolongation $\pr (Q \partial_u)$ satisfies
\begin{equation}\label{varsym}
\pr (Q \partial_u) L = \D_1 F_1 + \D_2 F_2
\end{equation}
for some functions $F_1,F_2: \cJ^\infty \rightarrow \R$. The pair $(F_1,F_2)$ is called the \emph{flux} of the variational symmetry.
%Here $\cD_{\pr (Q \partial_u)}$ denotes the Lie derivative with respect to the prolonged vector field ${\pr (Q \partial_u)}$.

A {\emph {conservation law}} for $L$ is a triple of functions $J_1,J_2,Q: \cJ^\infty \rightarrow \R$ that satisfy
\begin{equation}\label{conservation}
 \D_1 J_1 + \D_2 J_2 = -Q \var{}{L}{u} .
\end{equation}
If Equation (\ref{conservation}) holds true, the pair $J = (J_1,J_2)$ is called the \emph{conserved current} and $Q$ the \emph{characteristic} of the conservation law. On solutions of the Euler-Lagrange equations (\ref{varpde}) the conserved current $J$ is divergence-free, hence its name.

The famous Noether's theorem \cite{N} establishes a one-to-one correspondence between conservation laws and variational symmetries. 

\begin{theorem}
\label{thm-noether}
Let $Q \partial_u$ be a variational symmetry of $L$. Then
\begin{align}
    J_1[u] &= \sum_{I \not\ni t_2} \left( (\D_I Q) \var{}{L}{u_{It_1}} \right) + \frac{1}{2} \sum_I \D_2 \left( (\D_I Q) \var{}{L}{u_{It_1t_2}} \right) - F_1[u] , \label{noether1}\\
    J_2[u] &= \sum_{I \not\ni t_1} \left( (\D_I Q) \var{}{L}{u_{It_2}} \right) + \frac{1}{2} \sum_I \D_1 \left( (\D_I Q) \var{}{L}{u_{It_1t_2}} \right) - F_2[u],
    \label{noether2}
\end{align}
define the components of the conserved current of a conservation law, where the pair of functions $(F_1,F_2)$ is the flux, as in Equation \eqref{varsym}.

Conversely, given a conserved current $(J_1,J_2)$, Equations \eqref{noether1} and \eqref{noether2} define the flux $(F_1,F_2)$ of a variational symmetry.
\end{theorem}

Note that Equations \eqref{noether1} and \eqref{noether2} contain variational derivatives with respect to partial derivatives of $u$:
\[ \var{}{L}{u_{I}} = \sum_{\alpha,\beta \ge 0} (-1)^{\alpha+\beta} \D_1^\alpha \D_2^\beta \der{L}{u_{I t_1^\alpha t_2^\beta}}. \]

We also observe that  $J_1$ and $J_2$ can be alternatively written as
\begin{align*}
    J_1[u] &= \sum_{\alpha \geq 0} \left( (\D_1^\alpha Q) \var{}{L}{u_{t_1^{\alpha+1}}} \right) + \frac{1}{2} \sum_{\alpha \geq 0} \sum_{\beta \geq 0} \D_2 \left( (\D_1^\alpha \D_2^\beta Q) \var{}{L}{u_{t_1^{\alpha+1}t_2^{\beta+1}}} \right) - F_1[u], \\
    J_2[u] &= \sum_{\beta \geq 0} \left( (\D_2^\beta Q) \var{}{L}{u_{t_2^{\beta+1}}} \right) + \frac{1}{2} \sum_{\alpha \geq 0} \sum_{\beta \geq 0} \D_1 \left( (\D_1^\alpha \D_2^\beta Q) \var{}{L}{u_{t_1^{\alpha+1}t_2^{\beta+1}}} \right) - F_2[u],
\end{align*}

\begin{proof}[Proof of Theorem \ref{thm-noether}]
The key point of the proof consists in the integration by parts of
\[ \pr (Q \partial_u) L
= \sum_I (\D_I Q) \der{L}{u_I}, \]
i.e.,\@ to write it in the form
\[ \pr (Q \partial_u) L
= Q \var{}{L}{u} + \D_1(\cdots) + \D_2(\cdots). \]

To perform the full calculation, observe that
\[ \der{L}{u_I} =  \var{}{L}{u_I} + \D_1 \var{}{L}{u_{It_1}} + \D_2 \var{}{L}{u_{It_2}} + \D_1  \D_2 \var{}{L}{u_{It_1t_2}}, \]
hence
\begin{align*}
\pr (Q \partial_u) L
&= \sum_I (\D_I Q) \left( \var{}{L}{u_I} + \D_1 \var{}{L}{u_{It_1}} + \D_2 \var{}{L}{u_{It_2}} + \D_1  \D_2 \var{}{L}{u_{It_1t_2}} \right) \\
&= \sum_I \big( (\D_{It_1t_2} Q) + (\D_{It_2} Q) \D_1 + (\D_{It_1} Q) \D_2 + (\D_I Q) \D_1 \D_2 \big) \var{}{L}{u_{It_1t_2}} \\
&\quad + \sum_{I \not\ni t_2} \big( (\D_{It_1} Q) + (\D_I Q) \D_1 \big) \var{}{L}{u_{It_1}}  \\
&\quad + \sum_{I \not\ni t_1} \big( (\D_{It_2} Q) + (\D_I Q) \D_2 \big) \var{}{L}{u_{It_2}} 
+ Q \var{}{L}{u} ,
\end{align*}
where the last term would be a sum over all $I \not\ni t_1,t_2$, but only the empty multi-index $I = (0,0)$ satisfies this condition. The above equation can be simplified as
\begin{align*}
\pr (Q \partial_u) L
&= \sum_I \D_1 \D_2 \left( (\D_I Q) \var{}{L}{u_{It_1t_2}} \right) 
\\
&\quad 
+ \sum_{I \not\ni t_2} \D_1 \left( (\D_I Q) \var{}{L}{u_{It_1}} \right) + \sum_{I \not\ni t_1} \D_2 \left( (\D_I Q) \var{}{L}{u_{It_2}} \right) +  Q \var{}{L}{u} \\
&= \D_1 (J_1 + F_1) + \D_2 (J_2 + F_2) + Q \var{}{L}{u} .
\end{align*}
It follows that Equations \eqref{varsym} and \eqref{conservation} are equivalent. Hence if $Q \partial_u$ is a variational symmetry, then Equations \eqref{noether1}--\eqref{noether2} define a conserved current.
\end{proof}

\begin{example}
\label{ex-KdV-sym}
Consider again the KdV equation
\[ u_{12} = u_{1111} + 6 u_1 u_{11} \]
and its Lagrangian
\[
L[u]= \frac{1}{2} u_1 u_2 - u_1^3 - \frac{1}{2} u_1 u_{111}.
\]
As before, indices  denote derivatives with respect to the corresponding time variables, e.g.\@ $u_{12} = u_{t_1 t_2}$. We present two variational symmetries of this equation and their associated conservation laws:

\begin{enumerate}[$(a)$]
\item 
The generalized vector field $Q \partial_u$ with $Q[u] = u_1$ corresponds to a translation in the $t_1$-direction. Indeed,
\[
\pr(Q \partial_u) L = u_1 \der{L}{u} + u_{11} \der{L}{u_1} + u_{111} \der{L}{u_{11}} + u_{1111} \der{L}{u_{111}} + u_{12} \der{L}{u_2} \\
 = \D_1 L,
\]
hence $Q \partial_u$ is a variational symmetry with flux
\[
(F_1[u], F_2[u]) = (L[u],0). 
\]

Corresponding to this variational symmetry we find the conservation law
\[
-Q[u] \var{}{L}{u} = -u_1 ( -u_{12} + 6 u_1 u_{11} + u_{1111})  = \D_1 J_1 + \D_2 J_2,
\]
with
\begin{align*}
J_1[u] &= u_1 \var{}{L}{u_1} + u_{11} \var{}{L}{u_{11}}  + u_{111} \var{}{L}{u_{111}} - F_1[u] = -2 u_1^3 - u_1 u_{111} + \frac{1}{2} u_{11}^2, \\
J_2[u] &= u_1 \var{}{L}{u_2} - F_2[u] = \frac{1}{2} u_1^2.
\end{align*}
This in turn implies the conservation of momentum:
\[ \D_2 \int \frac{1}{2} u_1^2 \,\d t_1 = 0 . \]

\item The generalized vector field  $Q \partial_u$ with
\[ Q[u] = 10 u_1^3 + 5 u_{11}^2 + 10 u_1 u_{111} + u_{11111}. \]
Indeed,
\begin{align*}
\pr(Q \partial_u) L &= Q \der{L}{u} + (\D_1 Q) \der{L}{u_1} +  (\D_1^2 Q) \der{L}{u_{11}} +  (\D_1^3 Q) \der{L}{u_{111}} +  (\D_2 Q) \der{L}{u_2} \\
&= \D_1 F_1+ \D_2 F_2,
\end{align*}
with
\begin{align*}
F_1[u] &= -18 u_{1}^{5} - 15 u_{1}^{2} u_{11}^{2} - 45 u_{1}^{3} u_{111} + 5 u_{1}^{3} u_{2} + 4 u_{11}^{2} u_{111} - 18 u_{1} u_{111}^{2} - 4 u_{1} u_{11} u_{1111} \\
&\quad - 8 u_{1}^{2} u_{11111} - 10 u_{1} u_{11} u_{12} + \frac{5}{2} u_{11}^{2} u_{2} + 5 u_{1} u_{111} u_{2} + \frac{1}{2} u_{1111}^{2} - u_{111} u_{11111} \\
&\quad + \frac{1}{2} u_{11} u_{111111} - \frac{1}{2} u_{1} u_{1111111} + u_{111} u_{112} - u_{1111} u_{12} + \frac{1}{2} u_{11111} u_{2},
\\
F_2[u] &= \frac{5}{2} u_1^{4} + \frac{15}{2} u_1 u_{11}^{2} + 5 u_1^{2} u_{111} - \frac{1}{2} u_{111}^{2} + \frac{1}{2} u_{1} u_{11111}
 .
\end{align*}
The corresponding conservation law is
\[
-Q[u] \frac{\delta L}{\delta u} 
= \D_1 J_1 + \D_2 J_2,
\]
with
\begin{align*}
J_1[u] &= Q \var{}{L}{u_1} + (\D_1 Q) \var{}{L}{u_{11}}  + (\D_{11} Q)\var{}{L}{u_{111}} - F_1 \\
&= -12 u_{1}^{5} - 15 u_{1}^{2} u_{11}^{2} - 10 u_{1}^{3} u_{111} + u_{11}^{2} u_{111} - 2 u_{1} u_{111}^{2} - 6 u_{1} u_{11} u_{1111} \\
&\quad + 10 u_{1} u_{11} u_{12} - \frac{1}{2} u_{1111}^{2} - u_{111} u_{112} + u_{1111} u_{12}
\end{align*}
and 
\[
J_2[u] = Q \var{}{L}{u_2} - F_2[u] 
= \frac{5}{2} u_{1}^{4} - 5 u_{1} u_{11}^{2} + \frac{1}{2} u_{111}^{2} .
\]
\end{enumerate}
\end{example}

\section{Pluri-Lagrangian field theory}
\label{sec-pluri}

In this Section we briefly review the main concepts of pluri-Lagrangian field theory. For further details see \cite{S1,S2,SV}.

\subsection{Integrable hierarchies of PDEs}

One of the defining features of an integrable PDE is that it possesses an infinite amount of symmetries and, correspondingly, an infinite amount of conservation laws. These symmetries define a family of PDEs that commute with the original one. 

Let us illustrate the concept of commuting PDEs on the basis of our leading example.

\begin{example}
In Example \ref{ex-KdV-sym}$(b)$ we proved that the generalized vector field
$Q \partial_u$, with
\[
Q[u] = 10 u_1^3 + 5 u_{11}^2 + 10 u_1 u_{111} + u_{11111},
\]
is a variational symmetry of the KdV equation
\[
u_{12} = u_{1111} + 6 u_1 u_{11} = 0.
\]
If we introduce a third independent variable $t_3$, we can define the PDE
\[ u_3 = 10 u_1^3 + 5 u_{11}^2 + 10 u_1 u_{111} + u_{11111}, \]
which commutes with the KdV equation itself. This means that both ways of calculating the mixed derivative $u_{123}$ agree on solutions:
\begin{align*}
\D_3 u_{12} 
&= \D_3 (u_{1111} + 6 u_1 u_{11}) \\
&= 540 u_{1}^{2} u_{11}^{2} + 180 u_{1}^{3} u_{111} + 480 u_{11}^{2} u_{111} + 300 u_{1} u_{111}^{2} + 480 u_{1} u_{11} u_{1111} + 90 u_{1}^{2} u_{11111} \\
&\quad + 70 u_{1111}^{2} + 110 u_{111} u_{11111} + 56 u_{11} u_{111111} + 16 u_{1} u_{1111111} + u_{111111111} \\
&= \D_1 \D_2 \left(10 u_1^3 + 5 u_{11}^2 + 10 u_1 u_{111} + u_{11111} \right) \\
&= \D_1 \D_2 u_3.
\end{align*}
\end{example}

Since symmetries lead to commuting equations, a natural perspective on an integrable PDE is to consider it as one equation belonging to an infinite {\emph{integrable hierarchy}}, i.e., an infinite set of integrable PDEs such that any two systems in this set are compatible. Such hierarchies are usually generated by recursion operators or master symmetries \cite{newell1985solitons,dickey2003soliton,Olver}.

\subsection{Pluri-Lagrangian problems}

Let us focus  on $(1+1)$-dimensional PDEs. A finite number of equations from a hierarchy can be embedded in a higher-dimensional multi-time, where they share a common space direction, say $t_1 = x$, but each equation has its own time coordinate, $t_2,t_3,\ldots$. Formally, we can embed the whole hierarchy into an infinite-dimensional space in the same way.
 
In the classical variational description of $(1+1)$-dimensional PDEs, we integrate a Lagrange function over (an open subset of) the 2-dimensional space-time. A variational structure of a hierarchy of such PDEs should include the classical variational description of each individual equation, i.e., integration over a 2-dimensional subspace. Therefore, it is natural for the role of a Lagrange function to be played by a differential 2-form.

Let $\cL \in \Omega^2(\R^N)$ be a 2-form depending on the infinite jet of a smooth field $u: \R^N \rightarrow \R$, i.e.,
\begin{equation}\label{2form}
\cL[u] = \sum_{i < j} L_{ij}[u] \,\d t_i \wedge \d t_j,
\end{equation}
with $L_{ij}:\cJ^\infty \rightarrow \R$.
We say that $u$ solves the \emph{pluri-Lagrangian problem} for $\cL$ if for any 2-dimensional submanifold $\Gamma \subset \R^N$ and for any infinitesimal variation $v(t_1,\ldots,t_N) \partial_u$ of $u$, where $v:\R^N \rightarrow \R$ and all its derivatives vanish at the boundary of $\Gamma$, we have
\[ \frac{\d}{\d \varepsilon} \bigg|_{\varepsilon = 0} \int_\Gamma \cL[u + \varepsilon v]= 0 . \]
This can also be written as
\[ \int_\Gamma \pr(v \partial_u) \cL[u] = 0 ,\]
where the vertical vector field $\pr(v \partial_u) = \sum_I v_I \der{}{u_I}$ acts on the coefficients of $\cL[u]$, i.e.
\[ \pr(v \partial_u) \cL[u] = \sum_{i < j} \sum_I v_I \der{L_{ij}[u]}{u_I} \,\d t_i \wedge \d t_j . \]

The equations that characterize solutions to the pluri-Lagrangian problem are called \emph{multi-time Euler-Lagrange equations}. They were derived in \cite{S2} and state that, for all $i,j,k \in \{1,\ldots,N\}$, there holds:
\begin{align}
    \forall I \not\ni t_i,t_j : &\quad \var{ij}{L_{ij}}{u_I} = 0, \label{EL1} \\
    \forall I \not\ni t_i: &\quad \var{ij}{L_{ij}}{u_{It_j}} = \var{ik}{L_{ik}}{u_{It_k}} , \label{EL2} \\
     \forall I: &\quad \var{ij}{L_{ij}}{u_{It_it_j}} + \var{jk}{L_{jk}}{u_{It_jt_k}} + \var{ki}{L_{ki}}{u_{It_kt_i}} = 0 , \label{EL3}
\end{align}
where 
\begin{equation}\label{varder-general}
\var{ij}{L_{ij}}{u_{I}} = \sum_{\alpha,\beta \ge 0} (-1)^{\alpha+\beta} \D_i^\alpha \D_j^\beta \!\left( \der{L_{ij}}{u_{I i^\alpha j^\beta}} \right)
\end{equation}
is the variational derivative in the $(t_i,t_j)$-plane. Note that the multi-time Euler-Lagrange equations contain the classical Euler-Lagrange equations in each $(t_i,t_j)$-plane \eqref{EL1}, where derivatives with respect to other times are considered as additional components of the field, plus additional equations \eqref{EL2}--\eqref{EL3} coming from choices of $\Gamma$ that are not coordinate planes. 

In the present work, we will use a different property to recognize solutions to the pluri-Lagrangian problem. There is a remarkable relation between the pluri-Lagrangian problem and the property that the 2-form $\cL$ is closed on solutions $u$ to the hierarchy. In fact, this closedness property is often considered to be the fundamental property of the Lagrangian theory of integrable hierarchies \cite{LN1, LN2,LN3, LNQ, YLN, ALN, XLN}. 
When this point of view is taken, the term  ``Lagrangian multiform'' is more commonly used than  ``pluri-Lagrangian''. 

Here, we show that a slightly weaker property of the 2-form is a sufficient condition for a solution to the pluri-Lagrangian problem.

\begin{theorem}\label{thm-closed-implies-solution}
Consider a 2-form $\cL$ and a hierarchy of commuting PDEs
\begin{equation}\label{hierarchy} 
u_i = Q_i[u] \qquad i = 2,\ldots, N,
\end{equation}
with $Q_i:\cJ^\infty \rightarrow \R$.
If the exterior derivative of $\cL$ is constant up to a term that attains a double zero on solutions of
(\ref{hierarchy}), i.e.,\@ if 
\[ \d \cL = \gamma + \sum_{I,J} \sum_{i,j} \omega_{i,j}^{I,J} \D_I (u_i - Q_i) \D_J (u_j - Q_j) \]
for some $\cJ^\infty$-dependent 3-forms $\omega_{i,j}^{I,J}$ and a 3-form $\gamma$ that does not depend on $u$ or its derivatives, then all solutions $u: \R^N \rightarrow \R$ to the hierarchy $\eqref{hierarchy}$ also solve the pluri-Lagrangian problem for $\cL$.
\end{theorem}

Strictly speaking, the assumption that the PDEs \eqref{hierarchy} commute can be dropped from this theorem. If they do not commute then there will usually be no non-trivial solutions $u: \R^N \rightarrow \R$ to all PDEs simultaneously, so in this case the theorem would be of very limited relevance.

\begin{proof}[Proof of Theorem \ref{thm-closed-implies-solution}]
Let $u$ be a solution to the hierarchy and $\Gamma = \partial B$ a surface defined as the boundary of a 3-manifold $B$. It is sufficient to show that the pluri-Lagrangian property holds on such surfaces. Indeed,
without loss of generality we can require variations to be supported on small open subsets and for any sufficiently small open subset $\Gamma'$ of a given surface, one can find a 3-manifold such that $\Gamma'$ is contained in its boundary.

As a consequence of the assumption on $\cL$ there holds for any variation $v: \R^N \rightarrow \R$ that
\[ \pr(v \partial_u) \d \cL[u] = \frac{\d}{\d \varepsilon}\bigg|_{\varepsilon=0} \d \cL[u+\varepsilon v] = 0. \]
Therefore
\[ \frac{\d}{\d \varepsilon}\bigg|_{\varepsilon=0} \int_\Gamma \cL[u+\varepsilon v]
= \frac{\d}{\d \varepsilon}\bigg|_{\varepsilon=0} \int_B d \cL[u+\varepsilon v]
= 0 ,\]
hence the action integral over any surface $\Gamma$ is critical with respect to variations of $u$.
\end{proof}

There are strong indications that the existence of a pluri-Lagrangian structure is deeply connected to integrability. One such indication comes from within the theory: the multi-time Euler-Lagrange equations are highly overdetermined. Hence if nontrivial solutions exist, then we are dealing with a system with remarkable properties. Other indications are connections to different notions of integrability, including Hamiltonian formulations \cite{S1,SV} and Lax pairs \cite{sleigh2018variational}, even though these connections have not yet been studied in full detail.

Despite some recent discoveries, relatively few examples of pluri-Lagrangian hierarchies of PDEs are known. To our knowledge, the list is limited to the potential KdV equation \cite{S2} and several related hierarchies obtained as continuum limits from lattice equations \cite{vermeeren2018variational,vermeeren2019continuum}, as well as (a matrix-valued generalization of) the AKNS system \cite{sleigh2018variational}. The goal of this paper is to establish a construction of a pluri-Lagrangian 2-form for a given hierarchy of $(1+1)$-dimensional PDEs, assuming we know classical Lagrange functions for the individual equations. Furthermore, we will assume that the vector field associated to each of the PDEs is a variational symmetry for the Lagrangians of the rest of the hierarchy. This assumption can be thought of as the Lagrangian analogue to commuting Hamiltonian flows.

\section{From variational symmetries to a pluri-Lagrangian 2-form}
\label{sec-construction}

We will take $t_1 = x$ to be the space coordinate. Then we can take the coefficients $L_{1j}$ of the pluri-Lagrangian 2-form \eqref{2form} to be classical Lagrangians  for the individual equations of the hierarchy. However, the coefficients $L_{ij}$ with $i,j>1$ do not have an interpretation in a classical variational principle. It is not obvious under which conditions suitable $L_{ij}$ exist, such that the given hierarchy solves the pluri-Lagrangian problem for the 
2-form. Below we will give an answer to this question for a large class of Lagrangians.

For a hierarchy of evolutionary equations,
\begin{equation}\label{hierarchy2} 
u_i = Q_i(u_1,u_{11},\ldots) \qquad i = 2,\ldots, N,
\end{equation}
it is a reasonable assumption that the corresponding Lagrangians do not contain second or higher derivatives with respect to the time variable. Similarly, we will assume that the Lagrangian does not contain products of time-derivatives. Suppose we have a family of Lagrangians $L_{1i}$ for $i = 2, \ldots, N$ satisfying these assumptions:
\begin{equation} \label{L1i}
     L_{1i}[u] = p(u,u_1,u_{11},\ldots) u_i - h_i(u,u_1,u_{11},\ldots).
\end{equation}
Here $p$ and $h$ are two arbitrary functions of their arguments. In particular the term $p(u,u_1,u_{11},\ldots) u_i$ plays the role of a kinetic energy.
Note that we are not including mixed derivatives, $u_{1i}, u_{11i},\ldots$. This does not restrict generality, because if a Lagrangian depends linearly on such derivatives, then we can integrate by parts to get an equivalent Lagrangian of the form \eqref{L1i}. Furthermore, note that the factor $p(u,u_1,u_{11},\ldots)$ in the kinetic term of $L_{1i}[u]$ is the same for all $i$. This is a direct consequence of the multi-time Euler-Lagrange equations of type \eqref{EL2}.

The Euler-Lagrange equations \eqref{varpde} of the Lagrangians \eqref{L1i} will not be evolutionary. Instead we assume that the Euler-Lagrange equations are differential consequences of the hierarchy \eqref{hierarchy2}, i.e., equations of the form
\[ \cE_p( u_i - Q_i(u_1,u_{11},\ldots) ) = 0 ,\]
where $\cE_p$ is some differential operator, depending on the kinetic term of the Lagrangians. In the case of the KdV hierarchy we have $\cE_p = \D_1$, see Example \ref{ex-KdV-Lag}.

Assume that the prolonged vector fields $\cD_i = \pr (Q_i \partial_u)$, corresponding to the equations of the hierarchy, commute pairwise and are variational symmetries of the $L_{1j}$:
\begin{equation}\label{L1j-varsym}
\cD_i L_{1j} = \D_1 A_{ij} + \D_j B_{ij}
\end{equation}
for some functions $A_{ij}, B_{ij}: \cJ^\infty \rightarrow \R$.
If we consider only those terms that contain a $t_j$-derivative, what remains of Equation \eqref{L1j-varsym} is of the form
\[ \cD_i(p u_j) = \D_1 \bar{A}_{ij}(u,u_1,u_j,\ldots) + \D_j B_{ij}(u,u_1,u_{11},\ldots)  \]
for some function $\bar{A}_{ij}: \cJ^\infty \rightarrow \R$. This is an algebraic identity (as opposed to an equality on solutions), hence we can replace $t_j$-derivatives by new dependent variables, e.g.\@ $u_j$ by a field denoted by $u_t$. We find
\[ \cD_i (p u_t) = \D_1 \bar{A}_{ij}(u,u_1,u_t,\ldots) + \D_k B_{ij}(u,u_1,u_{11},\ldots). \]
Since the left hand side of this equation is independent of $j$, we can choose $\bar{A}_{ij}$ and $B_{ij}$ independent of $j$ as well. In particular, we can write $B_{ij} = B_i$ and get
\begin{equation}
\cD_i L_{1j} = \D_1 A_{ij} + \D_j B_{i}.
\end{equation}

Note that $A_{ij}, B_i: \cJ^\infty \rightarrow \R$ are only defined up to a constant, hence we can choose them to be zero on the zero field: $A_{ij}[0] = B_i[0] = 0$.

\begin{lemma}
\label{lemma}
For Lagrangians of the form \eqref{L1i} with commuting variational symmetries \eqref{L1j-varsym}, there exist functions $F_{ij}: \cJ^\infty \rightarrow \R:[u] \mapsto F_{ij}(u,u_1,u_{11},\ldots)$, that do not depend on any time-derivatives, such that
\begin{equation}\label{Fij}
\D_1 F_{ij} = \D_i L_{1j} - \D_j L_{1i} 
\end{equation}
on solutions of the hierarchy \eqref{hierarchy2}.
\end{lemma}
\begin{proof}
Since the variational symmetries $\cD_i = \pr (Q_i \partial_u)$ commute, we have for any $k \neq i,j$
\begin{align*}
0 &= [\cD_{i}, \cD_{j}] L_{1k} \\
&= \D_1 \left( \cD_{i} A_{jk} - \cD_{j} A_{ik} \right) + \D_k \left( \cD_{i} B_j - \cD_{j} B_i \right) .
\end{align*}
Now let $u$ be an arbitrary compactly supported smooth field. Then
\begin{align*}
0 &= \int_{-\infty}^\infty \D_1 \left( \cD_{i} A_{jk} - \cD_{j} A_{ik} \right) + \D_k \left( \cD_{i} B_j - \cD_{j} B_i \right) \d t_1 \\ 
&= \int_{-\infty}^\infty \D_k \left( \cD_{i} B_j - \cD_{j} B_i \right) \d t_1  \\ 
&= \D_k  \int_{-\infty}^\infty \left( \cD_{i} B_j - \cD_{j} B_i \right) \d t_1 .
\end{align*}
Since $u$ and in particular its $t_k$-derivatives are arbitrary, it follows that $\cD_{i} B_j - \cD_{j} B_i$ is a null Lagrangian. This implies (see e.g.\@ \cite[Theorem 4.7]{Olver}) that there exists a function $G_{ij}:\cJ^\infty \rightarrow \R$ such that
\[ \cD_{i} B_j - \cD_{j} B_i = \D_1(G_{ij}). \]
Hence with $F_{ij} = G_{ij} + A_{ij} - A_{ji}$ we find that, on solutions of the hierarchy \eqref{hierarchy2},
\begin{align*}
\cD_{i} L_{1j} - \cD_{j} L_{1i}
&=  \D_1 A_{ij} + \D_j B_i - \D_1 A_{ji} - \D_i B_j \\
&= \D_1 A_{ij} + \cD_{j} B_i - \D_1 A_{ji} - \cD_{i} B_j \\
&= \D_1 F_{ij}. 
\end{align*}
Since we are working on solutions of the equations the hierarchy, we can use those equations to eliminate time-derivatives from $F_{ij}$, hence we can assume it depends on the jet as $F_{ij}(u,u_1,u_{11},\ldots)$.
\end{proof}

We now present our main result, which is the analogue in 2-dimensional field theory of Theorem 10 in \cite{petrera2017variational}.

\begin{theorem}
\label{thm-construction}
Assume we have Lagrangians of the form \eqref{L1i} with commuting variational symmetries \eqref{L1j-varsym}. Let
\begin{align}
L_{ij}[u] 
&= \sum_{\alpha \geq 0} \var{1j}{L_{1j}}{u_{t_1^{\alpha+1}}} \D_1^\alpha (u_i - Q_i) - \sum_{\alpha \geq 0} \var{1i}{L_{1i}}{u_{t_1^{\alpha+1}}} \D_1^\alpha (u_j - Q_j) + F_{ij}(u,u_1,u_{11},\ldots) , 
\label{Lij}
\end{align}
where $F_{ij}: \cJ^\infty \rightarrow \R$ is as in Lemma \ref{lemma} and the operator $\var{ij}{}{}$ is the variational derivative from Equation \eqref{varder-general}.
Then every solution of the hierarchy \eqref{hierarchy2} is a critical point of
\[ \cL[u] = \sum_{i<j} L_{ij}[u] \,\d t_i \wedge \d t_j \]
in the pluri-Lagrangian sense.
\end{theorem}

\begin{proof}
We show that $\cL$ is almost-closed in the sense of Theorem \ref{thm-closed-implies-solution}. We start by calculating $\D_1 L_{ij}$. We have:
\begin{align*}
\D_1 & \left(\sum_{\alpha \geq 0} \var{1j}{L_{1j}}{u_{t_1^{\alpha+1}}} \D_1^\alpha (u_i - Q_i) \right) 
\\
&= \sum_{\alpha \geq 0}\D_1 \var{1j}{L_{1j}}{u_{t_1^{\alpha+1}}} \D_1^\alpha (u_i - Q_i) + \sum_{\alpha \geq 0} \var{1j}{L_{1j}}{u_{t_1^{\alpha+1}}} \D_1^{\alpha+1} (u_i - Q_i) \\
&= \sum_{\alpha \geq 0} \left( \D_1 \var{1j}{L_{1j}}{u_{t_1^{\alpha+1}}} + \var{1j}{L_{1j}}{u_{t_1^\alpha}} \right) \D_1^\alpha (u_i - Q_i)  - \var{1j}{L_{1j}}{u} \D_1 (u_i - Q_i)  \\
&= \sum_{\alpha \geq 0} \left( \der{L_{1j}}{u_{t_1^\alpha}} - \D_j \var{1j}{L_{1j}}{u_{t_1^{\alpha} t_j}} - \D_1 \D_j \var{1j}{L_{1j}}{u_{t_1^{\alpha+1} t_j}} \right) \D_1^\alpha (u_i - Q_i) - \var{1j}{L_{1j}}{u} \D_1 (u_i - Q_i) .
\end{align*}
Since $L_{1j}$ does not depend on any mixed derivatives $u_{t_1^{\alpha+1} t_j}$, this simplifies to
\begin{align*}
\D_1 & \left(\sum_{\alpha \geq 0} \var{1j}{L_{1j}}{u_{t_1^{\alpha+1}}} \D_1^\alpha (u_i - Q_i) \right) 
\\
&= \sum_{\alpha \geq 0} \der{L_{1j}}{u_{t_1^\alpha}} \D_1^\alpha (u_i - Q_i) - \D_j \var{1j}{L_{1j}}{u_{t_j}} (u_i - Q_i) - \var{1j}{L_{1j}}{u} \D_1 (u_i - Q_i) \\
&\equiv \sum_{\alpha \geq 0} \der{L_{1j}}{u_{t_1^\alpha}} \D_1^\alpha (u_i - Q_i) - (\D_j p) (u_i - Q_i),
\end{align*}
where $\equiv$ denotes equality modulo double zeros. Similarly, there holds
\[ \D_1 \left( \sum_{\alpha \geq 0} \var{1i}{L_{1i}}{u_{t_1^{\alpha+1}}} \D_1^\alpha (u_j - Q_j) \right) 
\equiv \sum_{\alpha \geq 0} \der{L_{1i}}{u_{t_1^\alpha}} \D_1^\alpha (u_j - Q_j) - (\D_i p) (u_j - Q_j) . \]
Hence
\begin{align*}
 \D_1 L_{ij} 
&\equiv \sum_{\alpha \geq 0} \der{L_{1j}}{u_{t_1^\alpha}} \D_1^\alpha (u_i - Q_i) 
- \sum_{\alpha \geq 0} \der{L_{1i}}{u_{t_1^\alpha}} \D_1^\alpha (u_j - Q_j) \\
&\quad - (\D_j p) (u_i - Q_i) + (\D_i p) (u_j - Q_j) + \D_1 F_{ij}.
\end{align*}
Using the assumption that the Lagrangians $L_{1i}$ and $L_{1j}$ are of the form \eqref{L1i}, we can write
\begin{align*}
&\D_i L_{1j} - \cD_{i} L_{1j} = p \D_j (u_i - Q_i) + \sum_{\alpha \geq 0} \der{L_{1j}}{u_{t_1^\alpha}} \D_1^\alpha (u_i - Q_i) \\
&\D_j L_{1i} - \cD_{j} L_{1i} = p \D_i (u_j - Q_j) + \sum_{\alpha \geq 0} \der{L_{1i}}{u_{t_1^\alpha}} \D_1^\alpha (u_j - Q_j) ,
\end{align*}
where $\cD_i = \pr(Q_i \partial_u)$ and $\cD_j = \pr(Q_j \partial_u)$.
Hence
\begin{equation}\label{closedness-1}
\begin{split}
\D_1 L_{ij} - \D_i L_{1j} + \D_j L_{1i} 
&\equiv - \cD_{i} L_{1j} + \cD_{j} L_{1i} - (\D_j p) (u_i - Q_i) - p \D_j (u_i - Q_i) \\
&\quad + (\D_i p) (u_j - Q_j) + p \D_i (u_j - Q_j) + \D_1 F_{ij} .
\end{split}
\end{equation}
By definition of $F_{ij}$ we have that $\D_1 F_{ij} - \cD_{i} L_{1j} + \cD_{j} L_{1i} = 0$ on solutions of \eqref{hierarchy2}. 
Furthermore, the only time derivatives in this expression come from the kinetic parts $p u_i$ and $p u_j$ of the Lagrangians. Therefore,
\begin{align}
\D_1 F_{ij} &- \cD_{i} L_{1j} + \cD_{j} L_{1i} \notag\\
&= - \cD_{i} (p u_j - p Q_j) + \cD_{j} (p u_i - p Q_i) \notag\\
&= - p \cD_{i} (u_j - Q_j) - (\cD_{i} p) (u_j - Q_j) + p \cD_{j} (u_i - Q_i) + (\cD_{j} p) (u_i - Q_i) \notag\\
&\equiv - p \D_i (u_j - Q_j) - (\D_i p) (u_j - Q_j) + p \D_j (u_i - Q_i) + (\D_j p) (u_i - Q_i)  . \label{closedness-4}
\end{align}
Combining Equations \eqref{closedness-1} and \eqref{closedness-4} gives
\begin{equation}\label{closedness-5}
\D_1 L_{ij} - \D_i L_{1j} + \D_j L_{1i} \equiv 0.
\end{equation}

Consider three copies of Equation \eqref{closedness-5}, each with an additional differentiation:
\begin{align*}
&\D_k ( \D_1 L_{ij} - \D_i L_{1j} + \D_j L_{1i} ) \equiv 0, \\
&\D_j ( \D_1 L_{ik} - \D_i L_{1k} + \D_k L_{1i} ) \equiv 0, \\
&\D_i ( \D_1 L_{jk} - \D_j L_{1k} + \D_k L_{1j} ) \equiv 0.
\end{align*}
A linear combination of these three equations gives us
\[ \D_1 ( \D_k L_{ij} - \D_j L_{ij} + \D_i L_{jk} ) \equiv 0. \]
Since all coefficients are autonomous, this implies that 
\begin{equation}\label{closedness-6}
\D_k L_{ij} - \D_j L_{ij} + \D_i L_{jk} \equiv \text{const} .
\end{equation}
Equations \eqref{closedness-5} and \eqref{closedness-6} together imply that $\cL$ fulfills the conditions of Theorem \ref{thm-closed-implies-solution}, hence every solution of the hierarchy \eqref{hierarchy2} is a critical point of the pluri-Lagrangian problem for $\cL$.
\end{proof}

Theorem \ref{thm-construction} and its proof are formulated for scalar systems, but they can be extended to the case of multicomponent systems. If $u = (u^1,\ldots,u^\ell)$ satisfies the equations $u^k_i = Q^k_i$, we construct the Lagrangian coefficients by 
\begin{align*}
L_{ij}[u]
&= \sum_{k = 1}^\ell \sum_{\alpha \geq 0} \var{1j}{L_{1j}}{u^k_{t_1^{\alpha+1}}} \D_1^\alpha \big( u^k_i - Q^k_i \big) - \sum_{k = 1}^\ell \sum_{\alpha \geq 0} \var{1i}{L_{1i}}{u^k_{t_1^{\alpha+1}}} \D_1^\alpha \big( u^k_j - Q^k_j \big) + F_{ij}(u,u_1,u_{11},\ldots) .
\end{align*}

\section{Examples}
\label{sec-examples}

In this last Section we discuss three examples. For the first one, the potential Korteweg-de Vries hierarchy, a pluri-Lagrangian structure is known in the literature \cite{SV}. Our discussion illustrates that this structure can be obtained using Theorem \ref{thm-construction}. The second example is the Nonlinear Schrödinger (NLS) hierarchy. Its pluri-Lagrangian structure can be considered as a special case of the one for the AKNS hierarchy obtained in \cite{sleigh2018variational}. The final example is the system consisting of the sine-Gordon and modified KdV equations, which indicates that the construction of Theorem \ref{thm-construction} can be adapted to non-evolutionary equations.

The calculations in this Section were performed in the SageMath software system \cite{SAGE}. The code is available at \cite{code}.

\subsection{Potential KdV hierarchy}

We start with our running example of the Korteweg-de Vries equation. The potential KdV hierarchy was the first complete hierarchy of PDEs for which a pluri-Lagrangian structure was found \cite{S2}. Here we show that this structure can also be derived using Theorem \ref{thm-construction}. We present only a minimal example consisting of just the first two equations in the hierarchy,
\begin{align}
u_2 &= 3 u_{1}^{2} + u_{111} , \label{kdv2} \\
u_3 &= 10 u_{1}^{3} + 5 u_{11}^{2} + 10 u_{1} u_{111} + u_{11111}. \label{kdv3}
\end{align}
The corresponding Lagrangians are
\begin{align*}
L_{12}[u] &= \frac{1}{2} u_{1} u_{2} -  u_{1}^{3} - \frac{1}{2} u_{1} u_{111}, \\
L_{13}[u] &= \frac{1}{2} u_{1} u_{3} - \frac{5}{2} u_{1}^{4} + 5 u_{1} u_{11}^{2} - \frac{1}{2} u_{111}^{2},
\end{align*}
and have as their Euler-Lagrange equations
\begin{align}
&\D_1( u_2 - (3 u_{1}^{2} + u_{111})) = 0, \label{kdv2-el} \\
&\D_1( u_3 - (10 u_{1}^{3} + 5 u_{11}^{2} + 10 u_{1} u_{111} + u_{11111})) = 0. \label{kdv3-el} 
\end{align}
On solutions of the evolutionary equations, there holds
\begin{align*}
\D_2 L_{13} - \D_3 L_{12} 
&= -10 u_{1}^{3} u_{12} + 10 u_{1} u_{11} u_{112} + 5 u_{11}^{2} u_{12} + 3 u_{1}^{2} u_{13} - uu_{111} u_{1112} + \frac{1}{2} u_{1} u_{1113} \\
&\qquad + \frac{1}{2} u_{111} u_{13} - \frac{1}{2} u_{13} u_{2} + \frac{1}{2} u_{12} u_{3} \\
&= 15 u_{1}^{4} u_{11} + 135 u_{1} u_{11}^{3} + 210 u_{1}^{2} u_{11} u_{111} + 25 u_{1}^{3} u_{1111} - 18u_{11} u_{111}^{2} \\
&\qquad + \frac{15}{2} u_{11}^{2} u_{1111} + 34 u_{1} u_{111} u_{1111} + 33 u_{1}u_{11} u_{11111} + \frac{13}{2} u_{1}^{2} u_{111111} \\
&\qquad + \frac{1}{2} u_{1111} u_{11111} - u_{111} u_{111111} + \frac{1}{2} u_{1} u_{11111111} .
\end{align*}
Integrating this gives us
\begin{align*}
F_{23}(u,u_1,u_{11},\ldots) &= 3 u_{1}^{5} + \frac{135}{2} u_{1}^{2} u_{11}^{2} + 25 u_{1}^{3} u_{111} - \frac{25}{2} u_{11}^{2} u_{111} + 7 u_{1} u_{111}^{2} + 20 u_{1} u_{11} u_{1111} \\
&\quad + \frac{13}{2} u_{1}^{2} u_{11111} + \frac{1}{2} u_{1111}^{2} - \frac{1}{2} u_{111} u_{11111} - \frac{1}{2} u_{11} u_{111111} + \frac{1}{2} u_{1} u_{1111111}.
\end{align*}
Let $Q_2$ and $Q_3$ be the right hand sides of Equations \eqref{kdv2} and \eqref{kdv3} . Then the remaining terms in Equation \eqref{Lij} are
\begin{align*}
&\var{13}{L_{13}}{u_1} ( u_2 - Q_2) = \left( \frac{1}{2} u_{3} - 10 u_{1}^{3} - 5 u_{11}^{2} - 10 u_{1} u_{111} - u_{11111}) \right) ( u_2 - 3 u_{1}^{2} - u_{111}), \\
&\var{13}{L_{13}}{u_{11}} \D_1 ( u_2 - Q_2) = (10 u_{1} u_{11} - u_{1111}) ( u_{12} - 6 u_{1} u_{11} - u_{1111}), \\
&\var{13}{L_{13}}{u_{111}} \D_{11} ( u_2 - Q_2) = u_{111} ( u_{112} - 6 u_{1} u_{111} - 6 u_{11}^2 - u_{11111}),
\end{align*}
and
\begin{align*}
&-\var{12}{L_{12}}{u_1} ( u_3 - Q_3) = -\left( \frac{1}{2} u_{2} - 3 u_{1}^{2} - u_{111} \right) ( u_3 - 10 u_{1}^{3} - 5 u_{11}^{2} - 10 u_{1} u_{111} - u_{11111}), \\
&-\var{12}{L_{12}}{u_{11}} \D_1 ( u_3 - Q_3) = -\frac{1}{2} u_{11} ( u_{13} - 30 u_{1}^{2} u_{11} - 20u_{11} u_{111} - 10 u_{1} u_{1111} - u_{111111}), \\
&-\var{12}{L_{12}}{u_{111}} \D_{11} ( u_3 - Q_3) \\
&\quad = \frac{1}{2} u_{1} ( u_{113} - 60 u_{1} u_{11}^{2} - 30 u_{1}^{2} u_{111} - 20 u_{111}^{2} - 30 u_{11} u_{1111} - 10 u_{1} u_{11111} - u_{1111111})).
\end{align*}
Adding everything together, as in Equation \eqref{Lij} of Theorem \ref{thm-construction}, we find
\begin{align*}
L_{23}[u] &= 3 u_{1}^{5} - \frac{15}{2} u_{1}^{2} u_{11}^{2} + 10 u_{1}^{3} u_{111} - 5 u_{1}^{3} u_{2} + \frac{7}{2} u_{11}^{2} u_{111} + 3 u_{1} u_{111}^{2} \\
&\quad - 6 u_{1} u_{11} u_{1111} + \frac{3}{2} u_{1}^{2} u_{11111}  + 10 u_{1} u_{11} u_{12} - \frac{5}{2} u_{11}^{2}u_{2} - 5 u_{1} u_{111} u_{2} \\
&\quad + \frac{3}{2} u_{1}^{2} u_{3} - \frac{1}{2} u_{1111}^{2} + \frac{1}{2} u_{111} u_{11111} - u_{111} u_{112} + \frac{1}{2} u_{1} u_{113} \\
&\quad + u_{1111} u_{12} - \frac{1}{2} u_{11} u_{13} - \frac{1}{2} u_{11111} u_{2} + \frac{1}{2} u_{111} u_{3} .
\end{align*}

Note that the classical Euler-Lagrange equations 
\[ \var{12}{L_{12}}{u} = 0 \qquad \text{and} \qquad \var{13}{L_{13}}{u} = 0 \]
yield Equations \eqref{kdv2-el}--\eqref{kdv3-el}, which are the $t_1$-derivatives of the potential KdV equations \eqref{kdv2}--\eqref{kdv3}. However, the multi-time Euler-Lagrange equations also contain the potential KdV equations themselves:
\begin{align*}
\var{12}{L_{12}}{u_1} = -\var{23}{L_{23}}{u_3} \qquad \Rightarrow \qquad&
\frac{1}{2} u_2 - 3 u_1^2 - u_{111} = -\frac{3}{2} u_1^2 - \frac{1}{2} u_{111}, \\
\var{13}{L_{13}}{u_1} = \var{23}{L_{23}}{u_2} \qquad \Rightarrow \qquad&
\frac{1}{2} u_3 - 10 u_1^3 - 5 u_{11}^2 - 10 u_1 u_{111} - u_{11111} \\
&\quad = - 5 u_1^3 - \frac{5}{2} u_{11}^2 - 5 u_1 u_{111} - \frac{1}{2} u_{11111}.
\end{align*}

\subsection{Nonlinear Schrödinger hierarchy}

The nonlinear Schrödinger equation is one of the most prominent integrable PDEs \cite{faddeev2007hamiltonian, kaup1978exact}. The corresponding hierarchy is discussed for example in \cite{olver1996trihamiltonian, anco2017integrable}. It is a special case of the AKNS hierarchy, the pluri-Lagrangian of which is studied in \cite{sleigh2018variational}. Here we construct a pluri-Lagrangian structure for the NLS hierarchy using Theorem \ref{thm-construction}.

In this example we consider a complex field $u: \R^N \rightarrow \C$. The first two equations of the hierarchy are the nonlinear Schr\"odinger equation itself and the complex modified KdV equation,
\begin{align}
u_{2} &= i u_{11} - 2 i |u|^2 u , \label{nls2} \\
u_{3} &= u_{111} - 6 |u|^2 u_{1} .  \label{nls3}
\end{align}
Fields $u$ that solve both these equations and their complex conjugates are critical fields for the Lagrangians (see e.g.\@ \cite{avan2016lagrangian})
\begin{align*}
L_{12}[u] &= \frac{i}{2} \left( u_{2} \bar u_{} -  u_{} \bar u_{2} \right) - |u_{1}|^2 - |u|^4 , \\
L_{13}[u] &=  \frac{i}{2} \left( u_{3} \bar u - u_{} \bar u_{3} \right) + \frac{i}{2} \left( u_{11} \bar u_{1} - u_{1} \bar u_{11} \right) + \frac{3i}{2} |u|^2 \left( u_{1} \bar u -  u \bar u_{1} \right) .
\end{align*}

For these Lagrangians Lemma \ref{lemma} gives us the function
\[ 
F_{23}(u,u_1,u_{11},\ldots) = 2 |u|^6 - \frac{3}{2} |u|^2 \left( u_{11} \bar u - u \bar u_{11} \right) - 6 |u u_1|^2  + \frac{1}{2} \left( u_{111} \bar u_{1} + u_{1} \bar u_{111} \right) + |u_{11}|^2 
\]
and Theorem \ref{thm-construction} provides the coefficient
\begin{align*}
L_{23}[u] &= -4 |u|^6 - u_{1}^2 \bar u_{}^2 - u_{}^2 \bar u_{1}^2 + 2 |u u_{1}|^2 + 2 |u|^2 \left( u_{11} \bar u + u \bar u_{11} \right) + \frac{3}{2} i |u|^2 \left( u_{2} \bar u - u \bar u_{2} \right) \\
&\quad + \frac{i}{2} \left( u_{12} \bar u_{1} - u_{1} \bar u_{12} \right) + u_{3} \bar u_{1} + u_{1} \bar u_{3} - |u_{11}|^2 + i \left( u_{11} \bar u_{2} - u_{2} \bar u_{11} \right) 
\end{align*}
of a pluri-Lagrangian 2-form \[\cL[u] = L_{12}[u] \,\d t_1 \wedge \d t_2 + L_{13}[u] \,\d t_1 \wedge \d t_3 + L_{23}[u] \,\d t_2 \wedge \d t_3.\]

Interestingly, in this example the classical Euler-Lagrange equations 
\[ \var{12}{L_{12}}{u} = 0 \qquad \text{and} \qquad \var{13}{L_{13}}{u} = 0 \]
already yield the evolutionary form of the NLS equations \eqref{nls2}--\eqref{nls3}. All other multi-time Euler-Lagrange equations, in particular those of the form
\[ 
\var{12}{L_{12}}{u_1} = -\var{23}{L_{23}}{u_3} \qquad \text{and} \qquad
\var{13}{L_{13}}{u_1} = \var{23}{L_{23}}{u_2} \]
are trivially satisfied.

\subsection{Sine-Gordon equation and modified KdV hierarchy}

Consider the sine-Gordon equation
\[ u_{12} = \sin u \]
and the (potential) modified KdV hierarchy
\begin{align*}
u_3 &= u_{111} + \frac{1}{2}u_1^3 , \\
u_4 &= \frac{3}{8} u_1^5 + \frac{5}{2} u_1 u_{11}^2 + \frac{5}{2} u_1^2 u_{111} + u_{11111} , \\
&\vdotswithin{=}
\end{align*}
This hierarchy consists of symmetries of the sine-Gordon equation (see, e.g.\@ \cite{olver1977evolution} or \cite[Section 5k]{newell1985solitons}).
The corresponding Lagrangians are
\begin{align*}
L_{12}[u] &= \frac{1}{2} u_{1} u_{2} -  \cos u , \\
L_{13}[u] &= \frac{1}{2} u_{1} u_{3} - \frac{1}{8} u_{1}^{4} + \frac{1}{2} u_{11}^{2}, \\
L_{14}[u] &= \frac{1}{2} u_{1} u_{4} - \frac{1}{16} u_{1}^6 - \frac{5}{12} u_{1}^3 u_{111} - \frac{1}{2} u_{111}^2 , \\
&\vdotswithin{=}
\end{align*}
Since the sine-Gordon equation is not evolutionary, Theorem \ref{thm-construction} does not apply to this hierarchy. Surprisingly, a naive adaptation of the construction leads to a suitable 2-form, at least for the first few equations of the hierarchy.

We start the construction of a pluri-Lagrangian 2-form in three dimensions, considering only $t_1$, $t_2$ and $t_3$.
Let $Q_3 = u_{111} + \frac{1}{2}u_1^3$. Then on solutions of the equations, there holds
\begin{align*}
\D_2 L_{13} - \D_3 L_{12} 
&= \frac{1}{2} u_{12} u_3 - \frac{1}{2} u_1^3 u_{12} + u_{11} u_{112} - \frac{1}{2} u_{13} u_2 - u_3 \sin u \\
&= -\frac{1}{2} u_{12} Q_3 - \frac{1}{2} u_2 \D_1 Q_3 - \frac{1}{2} u_1^3 \sin u + u_{11} u_1 \cos u \\
&= \D_1 F_{23}
\end{align*}
for
\[ F_{23}[u] =  -\frac{1}{2} u_2 \left(u_{111} + \frac{1}{2} u_1^3 \right) + \frac{1}{2} u_1^2 \cos u . \]
Since there is no evolutionary equation for $u_2$, we tolerate the dependence of $F_{23}$ on this derivative. For the same reason, the term $\var{12}{L_{12}}{u_{1^{\alpha+1}}} \D_1^\alpha (u_2 - Q_2)$ in Equation \eqref{Lij} only makes sense for $\alpha > 0$. For $\alpha = 0$ we just remove it. We are left with
\begin{align*}
L_{23}[u] &= \var{13}{L_{13}}{u_{11}} (u_{12} - \sin u) -  \var{12}{L_{12}}{u_{1}} (u_{3} - u_{111} - \frac{1}{2}u_1^3) + F_{23}[u] \\
&= u_{11} (u_{12} - \sin u) - \frac{1}{2} u_2 (u_{3} - u_{111} - \frac{1}{2}u_1^3)  -\frac{1}{2} u_2 \left(u_{111} + \frac{1}{2} u_1^3 \right) + \frac{1}{2} u_1^2 \cos u \\
&= u_{11} (u_{12} - \sin u) - \frac{1}{2} u_2 u_3 + \frac{1}{2} u_1^2 \cos u .
\end{align*}
This pluri-Lagrangian structure in $\R^3$ was first found in \cite{S2}, but a pluri-Lagrangian structure incorporating more equations of the hierarchy has not been given previously. With the method presented here, such an extension is obtained by a straightforward (but long) calculation. For example, we can calculate $F_{24}$ and $F_{34}$ analogously to $F_{23}$ above. This in turn allow us to calculate the coefficients of the Lagrangian 2-form,
\begin{align*}
    L_{24}[u] &= \frac{3}{8} u_{1}^4 \cos u - \frac{5}{12} u_{1}^3 u_{112} + \frac{5}{4} u_{1}^2 u_{11} u_{12} - \frac{3}{2} u_{1}^2 u_{11} \sin u  - \frac{1}{2} u_{11}^2 \cos u \\
    &\quad + u_{1} u_{111} \cos u - u_{111} u_{112} + u_{1111} u_{12} - \frac{1}{2} u_{2} u_{4} - u_{1111} \sin u
\end{align*}
and 
\begin{align*}
    L_{34}[u] &= \frac{3}{128} u_{1}^8 - \frac{5}{16} u_{1}^4 u_{11}^2 + \frac{7}{16} u_{1}^5 u_{111} - \frac{3}{16} u_{1}^5 u_{3} - \frac{1}{8} u_{11}^4 + \frac{7}{4} u_{1} u_{11}^2 u_{111} \\
    &\quad + \frac{3}{4} u_{1}^2 u_{111}^2 - \frac{3}{2} u_{1}^2 u_{11} u_{1111} + \frac{1}{4} u_{1}^3 u_{11111} - \frac{5}{12} u_{1}^3 u_{113} + \frac{5}{4} u_{1}^2 u_{11} u_{13}  \\
    &\quad - \frac{5}{4} u_{1} u_{11}^2 u_{3} - \frac{5}{4} u_{1}^2 u_{111} u_{3} + \frac{1}{4} u_{1}^3 u_{4} - \frac{1}{2} u_{1111}^2 + \frac{1}{2} u_{111} u_{11111}  \\
    &\quad - u_{111} u_{113} + u_{1111} u_{13} - u_{11} u_{14} - \frac{1}{2} u_{11111} u_{3} + \frac{1}{2} u_{111} u_{4} .
\end{align*}

The presented hierarchy can be extended to a doubly-infinite hierarchy, where the sine-Gordon equation connects two copies of the modified KdV hierarchy, one as stated above and one where $t_2$ is used as space variable. The calculations presented here can be easily extended to cover both sides of the hierarchy. A pluri-Lagrangian structure of this double hierarchy was previously obtained using a carefully chosen continuum limit \cite{vermeeren2018variational}. 

In this example, a straightforward adaptation of Equation \eqref{Lij} gives us suitable coefficients $L_{ij}$. However, there does not seem to be a simple generalization of the proof we gave for Theorem \ref{thm-construction} to cover this case. In this example we have verified by direct calculation that the multi-time Euler-Lagrange equations consist of the Sine-Gordon and modified KdV equations and differential consequences thereof. Showing the validity of our construction in a more general setting, ideally with a more conceptual proof, is a goal for future research.

\section{Conclusions}

We have shown that a hierarchy of 2-dimensional variational PDEs, that are variational symmetries of each other, possesses a pluri-Lagrangian  structure. This extends the results of \cite{petrera2017variational}, where a similar result was obtained for variational ODEs. The existence of a hierarchy of variational symmetries for a PDE is closely related to its integrability. Hence our result contributes significantly to the evidence that pluri-Lagrangian structures are a fundamental feature of integrability. Furthermore, our construction can be used to obtain new examples of pluri-Lagrangian 2-forms, as we illustrated in the context of the nonlinear Schrödinger hierarchy.

As illustrated by the example of the Sine-Gordon and mKdV equations, our construction applies more generally than the proof we provided. More research is needed to determine the most general form of the ideas presented here. Relevant to this line of investigation is the paper \cite{sleigh2019variational}, which deals with the same topics as the present work (and appeared on the arXiv one day after it).

\section*{Acknowledgments}

The authors are grateful to Yuri Suris for helpful discussions and feedback on a draft of this manuscript. 

The authors are partly supported by DFG (Deutsche Forschungsgemeinschaft) in the frame of SFB/TRR 109 ``Discretization in Geometry and Dynamics''.

%%%%%%%%%%%%%%%%%%%%%%
%%%%%%%%%%%%%%%%%%%%%%

%\bibliographystyle{abbrvnat_mv}
%\bibliography{bib}

\begin{thebibliography}{99}
%%%%%%%%%%%%%%%%%%%%%%
%%%%%%%%%%%%%%%%%%%%%%

\bibitem{anco2017integrable}
S. C. Anco, F. Mobasheramini,
{\em Integrable $U(1)$-invariant peakon equations from the NLS hierarchy},
Physica D 355, 1--23 (2017).


\bibitem{ALN}
J.~Atkinson, S.B.~Lobb, F.W.~Nijhoff, 
{\em An integrable multicomponent quad-equation and its
Lagrangian formulation}, 
Theor. Math. Phys. 173, No. 3, 1644--1653 (2012); translation from Teor. Mat. Fiz. 173, No. 3, 363--374 (2012).

\bibitem{avan2016lagrangian}
J. Avan, V. Caudrelier, A. Doikou, A. Kundu,
{\em Lagrangian and Hamiltonian structures in an integrable hierarchy and space-time duality},
Nucl. Phys., B 902, 415--439 (2016).



\bibitem{BS}
A.I. Bobenko, Yu.B. Suris,
{\em Discrete pluriharmonic functions as solutions of linear pluri-Lagrangian systems}, 
Commun. Math. Phys. 336, No. 1, 199--215 (2015).


\bibitem{BPS1}
R. Boll, M. Petrera, Yu.B. Suris, 
{\em Multi-time Lagrangian 1-forms for families of B\"acklund transformations. Toda-type systems},
J. Phys. A, Math. Theor. 46, No. 27, Article ID 275204, 26 p. (2013).


\bibitem{BPS2}
R. Boll, M. Petrera, Yu.B. Suris, 
{\em What is integrability of discrete variational systems?},
Proc. R. Soc. Lond., Ser. A, Math. Phys. Eng. Sci. 470, No. 2162, Article ID 20130550, 15 p. (2014).

\bibitem{BPS3}
R. Boll, M. Petrera, Yu.B. Suris, 
{\em Multi-time Lagrangian 1-forms for families of B\"acklund transformations. Relativistic Toda-type systems},
J. Phys. A, Math. Theor. 48, No. 8, Article ID 085203, 28 p. (2015).

\bibitem{BPS4}
R. Boll, M. Petrera, Yu.B. Suris, 
{\em On the variational interpretation of the discrete KP equation},
in: Advances in discrete differential geometry. Berlin: Springer. 379--405 (2016).

\bibitem{BPS5}
R. Boll, M. Petrera, Yu.B. Suris, 
{\em On integrability of discrete variational systems: octahedron relations},
Int. Math. Res. Not. 2016, No. 3, 645--668 (2016).


\bibitem{dickey2003soliton}
{L. A. Dickey}, 
{\em Soliton equations and Hamiltonian systems}, Singapore: World Scientific (1991).

\bibitem{faddeev2007hamiltonian}
L. D. Faddeev and L. A. Takhtajan,
{\em Hamiltonian methods in the theory of solitons},
Transl. from the Russian by A. G. Reyman. Reprint of the 1987 original. Berlin: Springer (2007).

\bibitem{kaup1978exact}
D. J. Kaup, A. C. Newell
{\em An exact solution for a derivative nonlinear Schrödinger equation},
J. Math. Phys. 19, 798--801 (1978)


\bibitem{LN1}
S. Lobb, F.W. Nijhoff,
\textit{Lagrangian multiforms and multidimensional consistency},
J. Phys. A, Math. Theor. 42, No. 45, Article ID 454013, 18 p. (2009).



\bibitem{LN2}
S.B. Lobb, F.W. Nijhoff,
{\em Lagrangian multiform structure for the lattice Gel'fand-Dikij hierarchy},
J. Phys. A, Math. Theor. 43, No. 7, Article ID 072003, 11 p. (2010).


\bibitem{LN3}
S.B. Lobb, F.W. Nijhoff,
{\em A variational principle for discrete integrable systems},
SIGMA, Symmetry Integrability Geom. Methods Appl. 14, Paper 041, 18 p. (2018).



\bibitem{LNQ}
S.B. Lobb, F.W. Nijhoff, G.R.W. Quispel,
{\em Lagrangian multiform structure for the lattice KP system}, 
J. Phys. A, Math. Theor. 42, No. 47, Article ID 472002, 11 p. (2009).

\bibitem{newell1985solitons}
{A. C. Newell}, 
{\em Solitons in mathematics and physics}, Philadelphia, PA: Society for Industrial and Applied Mathematics (SIAM) (1985).

\bibitem{N}
E. Noether,
{\em Invariante Variationsprobleme},
Nachr. Ges. Wiss. G\"ottingen, Math.-Phys. Kl. 1918, 235--257 (1918).


\bibitem{olver1977evolution}
P.J. Olver, 
{\em Evolution equations possessing infinitely many symmetries},
J. Math. Phys. 18, No. 6, 1212--1215 (1977).

\bibitem{Olver}
P.J. Olver,
{\em Applications of Lie groups to differential equations},
2nd ed. New York: Springer-Verlag (1993).

\bibitem{olver1996trihamiltonian}
P.J. Olver, P. Rosenau,
{\em Tri-Hamiltonian duality between solitons and solitary-wave solutions having compact support},
Phys. Rev. E 53, 1900--1906 (1996).


\bibitem{petrera2017variational}
M. Petrera, Yu. B. Suris,
{\em Variational symmetries and pluri-Lagrangian systems in classical mechanics},
J. Nonlinear Math. Phys. 24:sup1, 121--145 (2017).

\bibitem{SAGE}
SageMath, the Sage Mathematics Software System (Version 8.1), \url{https://www.sagemath.org} (2017).

\bibitem{sleigh2018variational}
D. G. Sleigh, F. W. Nijhoff, V. Caudrelier, 
{\em A Variational Approach to Lax Representations},
Journal of Geometry and Physics, 142, 66--79 (2019)

\bibitem{sleigh2019variational}
D. G. Sleigh, F. W. Nijhoff, V. Caudrelier, 
{\em Variational symmetries and Lagrangian multiforms}, arXiv:1906.05084 (2019).

\bibitem{S1}
Yu.B.~Suris,
{\em Variational formulation of commuting Hamiltonian flows: multi-time Lagrangian 1-forms},
J. Geom. Mech. 5, No. 3, 365--379 (2013).


\bibitem{S2}
Yu.B. Suris,
{\em Variational symmetries and pluri-Lagrangian systems},
in: Dynamical systems, number theory and applications. A Festschrift in honor of Armin Leutbecher's 80th birthday. Hackensack, NJ: World Scientific. 255--266 (2016).


\bibitem{SV}
Yu.B. Suris, M. Vermeeren, 
{\em On the Lagrangian structure of integrable hierarchies},
in: Advances in discrete differential geometry. Berlin: Springer. 347--378 (2016).

\bibitem{vermeeren2018variational}
M. Vermeeren,
{\em A variational perspective on continuum limits of ABS and lattice GD equations},
SIGMA, Symmetry Integrability Geom. Methods Appl. 15, Paper 044, 35 p. (2019).

\bibitem{vermeeren2019continuum}
M. Vermeeren,
{\em Continuum limits of pluri-Lagrangian systems},
Journal of Integrable Systems 4, No. 1, Article ID xyy020 (2019).

\bibitem{code}
M. Vermeeren, {\em Support code for ``Variational symmetries and pluri-Lagrangian structures for integrable hierarchies of PDEs''}, \href{https://doi.org/10.5281/zenodo.3243313}{DOI: 10.5281/zenodo.3243313}

\bibitem{XLN}
P. Xenitidis, S. Lobb, F. Nijhoff, 
{\em On the Lagrangian formulation of multidimensionally consistent systems},
Proc. R. Soc. Lond., Ser. A, Math. Phys. Eng. Sci. 467, No. 2135, 3295--3317 (2011).



\bibitem{YLN} S. Yoo-Kong, S. Lobb, F. Nijhoff,
{\em Discrete-time Calogero-Moser system and Lagrangian 1-form structure},
J. Phys. A, Math. Theor. 44, No. 36, Article ID 365203, 39 p. (2011).





\end{thebibliography}
\end{document}